%% file: main.tex
\documentclass[11pt, a4paper]{article}

\usepackage[margin=1in]{geometry}
\usepackage{charter}
\usepackage{extarrows}
\usepackage{amsthm}
\allowdisplaybreaks
\usepackage{amssymb}
\usepackage{enumitem}
\usepackage[linktoc=page]{hyperref}
\hypersetup{
    colorlinks,
    linkcolor={red},
    citecolor={blue},
    urlcolor={OliveGreen}
}

\newtheorem{lemma}{Lemma}
\newtheorem{theorem}{Theorem}

\title{Optimal Budget-Feasible Mechanisms for Additive Valuations\footnote{Work
		done in part when the second and the last authors were research assistants at ITCS, Shanghai University of Finance
		and Economics.}}
\author{
Nick Gravin\thanks{ITCS, Shanghai University of Finance and Economics. {\tt
nikolai@mail.shufe.edu.cn}.} \and
Yaonan Jin\thanks{Department of Computer Science, Columbia University. {\tt jin.yaonan@columbia.edu}.} \and
Pinyan Lu\thanks{ITCS, Shanghai University of Finance and Economics. {\tt
lu.pinyan@mail.shufe.edu.cn}.} \and
{Chenhao Zhang}\thanks{Department of EECS, Northwestern University. {\tt chenhao.zhang.rea@u.northwestern.edu}}
}

\date{}

\input{notation}

\begin{document}
\maketitle

\begin{abstract}
In this paper, we show a tight approximation guarantee for budget-feasible mechanisms with an additive buyer. We propose a new simple randomized mechanism with approximation ratio of $2$, improving the previous best known result of $3$. Our bound is tight with respect to either the optimal offline benchmark, or its fractional relaxation. We also present a simple deterministic mechanism with the tight approximation guarantee of $3$ against the fractional optimum, improving the best known result of $(2+ \sqrt{2})$ for the weaker integral benchmark.
\end{abstract}

\maketitle

\section{Introduction}
\label{sec:intro}
\input{intro}

\subsection{Related Work}
\label{sec:related}
\input{related}

\section{Preliminaries}
\label{sec:prelim}
\input{prelim}

\section{Composition of Mechanisms: Pruning}
\label{sec:pruning}
\input{pruning}

\section{Deterministic Mechanism}
\label{sec:det}
\input{deterministic}

\section{Main Result: Randomized Mechanism}
\label{sec:rdm}
\input{randomized}

\section{Conclusion and Open Question}
\label{sec:open}
\input{open}

\vspace{.1in}
\noindent
{\bf Acknowledgements.}
We are grateful to the anonymous reviewers for their dedication in carefully reading through this paper; they offered many invaluable comments and suggestions.

% Bibliography
\bibliographystyle{alpha}
\bibliography{main}

\end{document}

%% file: notation.tex
\newcommand{\AutoAdjust}[3]{\mathchoice{ \left #1 #2  \right #3}{#1 #2 #3}{#1 #2 #3}{#1 #2 #3} }
\newcommand{\Xcomment}[1]{{}}

\newcommand{\InBrackets}[1]{\AutoAdjust{[}{#1}{]}}% {\left[{#1}\right]}

\newcommand{\Prx}[2][]{\operatorname{Pr}_{#1}\InBrackets{#2}}

\newcommand{\be}{\begin{equation}}
\newcommand{\ee}{\end{equation}}

\newcommand{\eqdef}{\stackrel{\textrm{\tiny def}}{=}}

\newcommand{\R}{\mathbb{R}}

\newcommand{\bid}{b}
\newcommand{\bids}{{\mathbf \bid}}
\newcommand{\bidsmi}{{\mathbf \bid}_{-i}}
\newcommand{\bidi}[1][i]{{\bid_{#1}}}

\newcommand{\cost}{c}
\newcommand{\costs}{{\mathbf \cost}}

\newcommand{\costi}[1][i]{{\cost_{#1}}}

\newcommand{\val}{v}
\newcommand{\vals}{{\mathbf \val}}

\newcommand{\vali}[1][i]{{\val_{#1}}}

\newcommand{\util}{u}

\newcommand{\utili}[1][i]{\util_{#1}}

\newcommand{\price}{p}
\newcommand{\prices}{{\mathbf \price}}

\newcommand{\alloc}{x}
\newcommand{\allocs}{{\mathbf \alloc}}

\newcommand{\alloci}[1][i]{{\alloc_{#1}}}

\newcommand{\opt}{\mathsf{opt}}
\newcommand{\fopt}{\mathsf{fopt}}
\newcommand{\alg}{\mathsf{alg}}

%% file: intro.tex
In a typical procurement setting, a buyer wants to purchase items from a set $A$ of agents. Each agent $i \in A$ can supply an item (or provide a service) at an incurred cost of $c_i$ to himself, and the buyer wants to optimize his valuation for the set of acquired items taking into account the costs of items. Because the agents may strategically report their costs, this setting is usually considered as a truthful mechanism design problem.

These problems have been extensively studied by the AGT community. The earlier work analyzed the case where the buyer's valuation takes $0$-$1$ values (see, e.g., \cite{ArcherT02}) in the frugality framework, with the objective of payment minimization. A more recent line of work on the budget-feasible mechanism design (see, e.g., \cite{S10}) studies more general valuation functions with a budget constraint of $B$ on the buyer's total payment. Our work belongs to the latter category.

Research in the budget-feasible framework focuses on different classes of complement-free valuations (ranging from the class of additive valuations to the most general class of subadditive valuations), and has many applications such as procurement in crowdsourcing markets~\cite{SM13}, experimental design~\cite{HIM14}, and advertising in social networks~\cite{S12}. 
%For example, consider an online crowdsourcing platform (like Amazon's Mechanical Turk), where a buyer may hire multiple workers to perform certain tasks like image labeling, text translations, or writing consumer surveys. The central problem for these online labor markets is to properly price each task, because the buyer is usually budget-constrained, and would run out of the budget once he pays too high to the hired workers.
The central problem for these online labor markets is to properly price each task. The budget feasibility mechanism design model is a very reasonable model that naturally captures the budget limitation on the buyer and also uncertainty about workers costs.

This setting corresponds to the most basic additive valuation of the buyer, which is the topic of our paper.
I.e., we assume that every hired worker $i \in W$ generates a value of $v_i \ge 0$ to the buyer, whose total valuation from all the hired workers $W$ is equal to $v(W) = \sum_{i \in W} v_i$. Without any incentive constraints, this naturally defines the {\sf Knapsack} optimization problem:
\begin{center}
    Find workers $S \subseteq A$: \quad\quad $\max_{S \subseteq A} v(S) = \sum_{i \in S} v_i$, subject to $\sum_{i \in S} c_i \leq B$.
\end{center}
In the budget-feasible framework, the goal is to design truthful direct-revelation mechanisms\footnote{Typically, there are no assumptions in the literature about the prior distribution of the agents' costs. The truthfulness condition means that the strategy of reporting the true cost is ex post a dominant and individually rational strategy for every single agent.} that decide (1)~which workers $W \subseteq A$ to select and (2)~how much to pay them under the budget constraint. A mechanism is evaluated against the benchmark of the optimal solution to the {\sf Knapsack} problem. Over all possible choices of the value $v_{i}$'s and the cost $c_{i}$'s, the worst-case multiplicative gap between the outcome $v(W)$ and the optimal {\sf Knapsack} solution is called the {\em approximation ratio} of this mechanism.

For the above problem with an additive buyer, Singer~\cite{S10} gave the first $5$-approximation mechanism. Later, the result was improved by Chen et al.~\cite{CGL11} via a $(2+ \sqrt{2})$-approximation deterministic mechanism and a $3$-approximation randomized mechanism, which still remain the best known upper bounds for the problem for nearly a decade. Further, the best known lower bounds are $(\sqrt{2} + 1)$ for the deterministic mechanisms and $2$ for the randomized ones~\cite{CGL11}. Thus, there are gaps for both the deterministic mechanisms $[\sqrt{2} + 1, 2+ \sqrt{2}]$ and the randomized ones $[2, 3]$. Since these two intervals intersect, it is even unclear whether the best randomized mechanism is indeed better than the best deterministic one.

Also for the above problem with an additive buyer, Anari et al.~\cite{AGN14} studied an important special case of large markets (i.e., the setting where each worker has vanishingly small cost compared to the buyer's budget) and acquired the tight bound of $\frac{e}{e - 1}$.

\vspace{.1in}
\noindent
{\bf Fractional Knapsack.}
Interestingly, all previous work on budget-feasible mechanisms for an additive buyer actually obtained results against the stronger benchmark of the optimal solution to \textsf{Fractional Knapsack}, i.e., the fractional relaxation of the \textsf{Knapsack} problem. (Nonetheless, the lower bounds apply to the \textsf{Knapsack} benchmark instead of the \textsf{Fractional Knapsack} benchmark.) Indeed, although \textsf{Knapsack} is a well-known NP-hard problem, its fractional relaxation admits an efficient solution by a simple greedy algorithm, and generally has much better behavior than the integral optimum. We also compare the performance of our mechanisms to the \textsf{Fractional Knapsack} benchmark.

\vspace{.1in}
\noindent
{\bf Our Results.}
We propose two natural mechanisms that both achieve tight guarantees against the \textsf{Fractional Knapsack} benchmark. Namely, we prove a $3$-approximation guarantee for a deterministic mechanism and a $2$-approximation guarantee for a randomized one. Given the matching lower bound of $2$ even against the weaker \textsf{Knapsack} benchmark, the guarantee from our randomized mechanism is also tight against the standard benchmark. Our results establish a clear separation between the respective power of randomized and deterministic mechanisms: no deterministic mechanism has an approximation guarantee better than $(\sqrt{2} + 1)$, whereas our randomized mechanism already achieves a $2$-approximation.

Concretely, we propose a new natural design principle of two-stage mechanisms. In the first stage, we greedily exclude the items with low {\em value-per-cost} ratios.\footnote{This is essentially the main approach used in the previous work, had we continued until the remaining items (as a whole) become budget-feasible.} Then in the second stage, we leverage the simple {\em posted-price} schemes, based on the values of the remaining items. Both of our randomized and deterministic mechanisms share the first stage, which stops earlier than its analogues from the previous work. A remarkable property of the first stage, which we call {\em pruning} (similar to the pruning approach in the frugality literature~\cite{ChenEGP10,KempeSM10}) is that, it can be {\em composed} (in the sense of~\cite{AggarwalH06}) with any truthful follow-up mechanism that runs on the items left to the second stage. The difference between our randomized and deterministic mechanisms lies in the follow-up posted-price schemes -- the randomized mechanism uses non-adaptive posted prices with the total sum below the budget, whereas our deterministic mechanism employs adaptive pricing that depends on whether the previous agents accepted or rejected their posted-price offers.

\vspace{.1in}
\noindent
{\bf Intuition behind our mechanism.}
The pruning stage of both mechanisms allows the buyer to reduce the choice complexity, and gives a reasonable upper bound on the payment to each remaining agent. The value of the fractional optimum never decreases too much, especially when the individual true cost $c_{i}$ of each remaining agent is a non-negligible fraction of the budget $B$. We prove that the fractional optimum drops at most by a factor of two after the pruning stage for an arbitrary set of values and costs.

The idea behind the pruning stage is that the removed agents can be safely ignored by the mechanism, since the remaining items suffice to get the desired approximation to the fractional optimum. Moreover, the mechanism should naturally prefer the items with higher value-per-cost ratios. Our pruning process is based on the value-per-cost ratio, and works specifically for an additive-valuation buyer. That is, it is still unknown how to extend such a pruning stage to more general classes of valuation functions.

The second stage of our randomized mechanism draws a random vector of budget-feasible posted prices. This is the same type of the mechanism as was used by Bei et al.~\cite{BCGL17} to establish the tight approximation ratio of $2$ for a subadditive buyer in the promise version of the problem (i.e., where the buyer is ensured to have a budget higher than the total cost of all items). Their result holds in the Bayesian setting, which by the {\em minimax principle} implies the existence of a randomized posted-price mechanism with the same approximation ratio in the worst-case setting. In our problem with an additive buyer, we explicitly construct a desired distribution over the posted-price vectors. Such posted-price schemes seem to be useful and easily adaptable to more general classes of valuation functions.

%% file: related.tex
%Frugality
A complementary concept of budget-feasible mechanism design is {\em frugality}, for which the objective is payment minimization under the feasibility constraint on the set of winning agents. In that framework, there is a rich literature studying different systems of feasible sets, including matroid set systems~\cite{KarlinKT05}, path and $k$-paths auctions~\cite{ArcherT02,Talwar03,ElkindSS04,ChenK07,ChenEGP10}, vertex cover and $k$-vertex cover~\cite{ElkindGG07a,KempeSM10,HajiaghayiKS18}.

The framework of budget-feasible mechanism design was proposed by Singer~\cite{S10}. Beyond additive valuations, other more general classes of complement-free valuations also have been considered in the literature:
\begin{center}
    submodular \quad $\subset$ \quad fractionally subadditive \quad $\subset$ \quad subadditive.
\end{center}
Singer gave an $112$-approximation mechanism for submodular valuations~\cite{S10}. This bound was improved to $7.91$ and $8.34$ respectively for the randomized and deterministic mechanisms by Chen et al.~\cite{CGL11}, and then to $4$ and $5$ by Jalaly and Tardos in~\cite{KT18}. For fractionally subadditive valuations, Bei et al.~\cite{BCGL17} gave a $768$-approximation randomized mechanism. For subadditive valuations, Dobzinski et al.~\cite{DPS11} first gave an $\mathcal{O}(\log^2 n)$-approximation randomized mechanism and an $\mathcal{O}(\log^3 n)$-approximation deterministic mechanism. Later, Bei et al.~\cite{BCGL17} showed the existence of an $\mathcal{O}(1)$-approximation mechanism in this most general setting. Nonetheless, an explicit description of such a mechanism is still unknown.

% Other variants
There also have been many interesting and practically motivated adjustments to the original budget feasibility model. In particular, Anari et al.~\cite{AGN14} investigated the variant with the additional large market assumption (namely, every agent has a negligible cost compared to the whole budget) and attained the tight result of $\frac{e}{e-1}$ for an additive buyer. Leonardi et al.~\cite{LMSZ17} explored an additive-valuation model where the winning agents must form an independent set from a matroid. Amanatidis et al.~\cite{ABM16,ABM17} investigated the variants with several important subclasses of submodular and fractionally subadditive valuations. Badanidiyuru et al.~\cite{BadanidiyuruKS12} studied the family of online pricing mechanisms in the budget feasibility model, motivated by practical restrictions given by the existing platforms. Balkanski and Hartline~\cite{BalkanskiH16} obtained improved guarantees in the Bayesian framework. Goel et al.~\cite{GoelNS14} concerned more complex scenarios on a crowdsourcing platform, where the buyer hires the workers to complete more than one task. Balkanski and Singer~\cite{BalkanskiS15} considered fair mechanisms (instead of truthful mechanisms) in the budget feasibility model.

%% file: prelim.tex
In the procurement auction, there are $n$ items for sale, each held by a single agent $i \in [n]$ with a privately known cost $\costi \geq 0$ and a publicly known value $v_i > 0$ for the buyer. The buyer has an {\em additive} valuation function $v(A)= \sum_{i \in A} v_i$ for purchasing a subset $A \subseteq [n]$ of items. Due to the revelation principle, we only consider direct-revelation mechanisms. Upon receiving {\em bids} $\bids = (\bidi)_{i = 1}^n$ of the claimed costs from the agents, a mechanism determines a set $W \subseteq [n]$ of {\em winning} agents and the {\em payments} $\prices = (p_i)_{i = 1}^n$ to the agents.

In the budget feasibility model, a {\em deterministic} mechanism $\mathcal{M}$ is specified by an allocation function $\allocs(\bids): \R_{+}^{n} \to \{0, 1\}^{n}$ (thus the winning set $W \eqdef \{i \in [n] ~ | ~ \alloci(\bids) = 1\}$) and a payment function $\prices(\bids): \R_{+}^{n} \to \R_{+}^n$. We use the notation $\bidi$ to denote the $i$-th entry of the bid vector $\bids$, and the notation $\bidsmi$ the bid vector without bidder $i \in [n]$. We are interested in those truthful mechanisms that satisfy the following properties for any $\bids = (\bidi)_{i = 1}^n$ and any $\costs = (\costi)_{i = 1}^n$.
\begin{itemize}
    %\item {\em Individual rationality} and {\em no positive transfer}: $p_i(\bids) \geq c_i$ if $i \in W$, and $p_i(\bids) = 0$ if $i\notin W$. Namely, each winning agent $i \in W$ gets a non-negative utility $\utili(\bids) = p_i(\bids) - \costi\ge 0$, and each losing agent $i \notin W$ gets a zero utility $\utili(\bids) = 0$.
    
    \item {\em Individual rationality}: $p_i(\bids) \geq c_i$ and thus $\utili(\bids) = p_i(\bids) - \costi \ge 0$ for every $i \in W$, while $p_i(\bids) \geq 0$ and thus $\utili(\bids) = p_i(\bids) \geq 0$ for every $i \notin W$. Namely, every agent $i \in [n]$ gets a non-negative utility.
    
    \item {\em Budget feasibility}: the total payment $\sum_{i \in W} p_i(\bids)$ is capped with a given budget $B \in \R_{+}$.
    
    \item {\em Truthfulness}: every agent $i \in [n]$ maximizes his utility when he bids the true cost $b_i = c_i$, namely $\utili(\costi, \bidsmi) \geq \utili(\bidi, \bidsmi)$ for any $c_i$ and any $\bids = (\bidi, \bidsmi)$.
\end{itemize}
It is well known (see~\cite{M81}) that truthfulness holds if and only if: (1)~the allocation function $\alloci(\bidi, \bidsmi)$ is monotone in bid $b_i$, i.e., each winning agent $i \in W$ keeps winning when he unilaterally claims a lower bid $b_i \le c_i$; and (2)~the payment $p_i(\bids)$ to each winning agent $i \in W$ is the threshold/maximum bid for him to keep winning, i.e., $p_i(\bids) = \sup \{\bidi \in \R_+ ~ | ~ \alloci(\bidi, \bidsmi) = 1\}$.

In general, a mechanism can have {\em randomized} allocation and payment rules. We restrict our attention to the mechanisms that can be described as a probability distributions over truthful deterministic mechanisms. Namely, any realization of such a randomized mechanism is some deterministic truthful mechanism that satisfies the above properties. A randomized mechanism of this type is called a {\em universally truthful} mechanism. We notice that most of the previous work on budget feasible mechanism only studies universally truthful mechanisms.

We denote by $\alg$ the value $\sum_{i \in W} v_i$ derived from a deterministic mechanism, or the expected value $\mathbb{E}\big[\sum_{i \in W} v_i\big]$ in case of a randomized mechanism. W.l.o.g., we assume $c_i \le B$ for each agent $i \in [n]$, since this agent cannot win when $c_i > B$ (due to the individual rationality and the budget feasibility constraints). If the buyer knows the private costs $\costs = (\costi)_{i = 1}^n$ of the items, he would simply select the subset of items with the maximum total value, under the budget constraint. Let $\opt$ denote the optimal solution to this \textsf{Knapsack} problem:
\[
\tag{\textsf{Knapsack}}
\opt \eqdef \max_{(x_i)_{i = 1}^n \in \{0, 1\}^n} \sum_{i = 1}^n x_i \cdot v_i,
\quad\mbox{ subject to } \sum_{i = 1}^n x_i \cdot c_i \leq B. \quad\quad\quad\quad
\]
We also consider the fractional relaxation of the problem, and define its optimum as
\[
\tag{\textsf{Fractional Knapsack}}
\fopt \eqdef \max_{(x_i)_{i = 1}^n \in [0, 1]^n} \sum_{i = 1}^n x_i \cdot v_i,
\quad \mbox{ subject to } \sum_{i = 1}^n x_i \cdot c_i \leq B.
\]
Although $\opt$ is NP-hard to calculate, finding $\fopt$ is easy: one greedily and divisibly takes the items in the decreasing order of their value-per-cost ratios,\footnote{Namely, the decreasing order $(\sigma_i)_{i = 1}^n$ is a permutation of $[n]$ such that $\frac{v_{\sigma_1}}{c_{\sigma_1}} \geq \frac{v_{\sigma_2}}{c_{\sigma_2}} \geq \cdots \geq \frac{v_{\sigma_n}}{c_{\sigma_n}}$.} until the budget is exhausted or no item is left. Under our assumption that $c_i\le B$ for all $i \in [n]$, we have $1 \le \frac{\fopt}{\opt} \le 2$.\footnote{Without this assumption, the gap between the two optima $\frac{\fopt}{\opt}$ can be arbitrary large.}

We say that a mechanism achieves an $\alpha$-approximation against the benchmark $\opt$, if under whatever values $\vals = (\vali)_{i = 1}^{n}$ and costs $\costs = (\costi)_{i = 1}^{n}$, the outcome value $\alg$
is at least an $\frac{1}{\alpha}$-fraction of the \textsf{Knapsack} solution $\opt$. In what follows, we usually evaluate a mechanism against the stronger benchmark $\fopt$, i.e., the solution to the \textsf{Fractional Knapsack} problem.
\[
    \alpha \le \max_{\vals, \costs, B} \frac{\opt}{\alg}
    \quad \Leftarrow \quad
    \alpha \le \max_{\vals, \costs, B} \frac{\fopt}{\alg}.
\]

%% file: pruning.tex
Every mechanism presented in this work can be described as a composition of two stages. In particular, all of our mechanisms share the same first stage, called {\sc Pruning-Mechanism}, which serves to exclude the items with low value-per-cost ratios.

\begin{figure}[h]
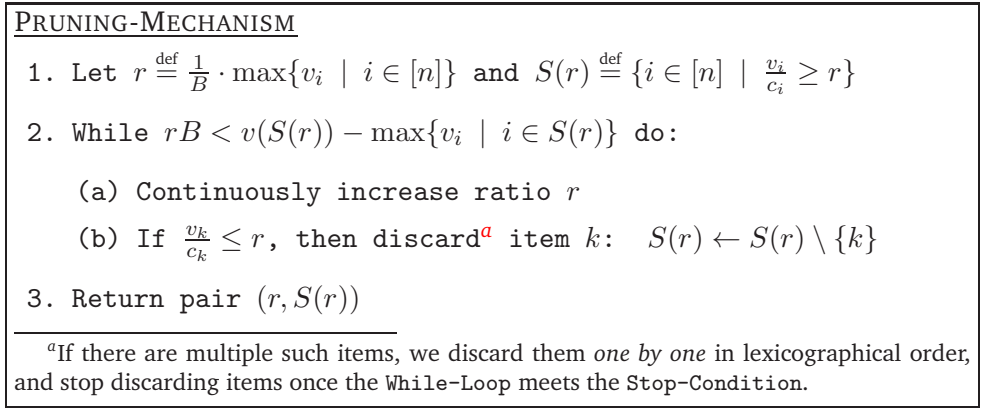

    \centering
    \fbox{\begin{minipage}{0.79\textwidth}
    \underline{\sc Pruning-Mechanism}
    \vspace{-.5em}
    {\tt \begin{enumerate}[itemindent=-0.5em]
        \item Let $r \eqdef \frac{1}{B} \cdot \max \{v_i ~ | ~ i \in [n]\}$ and $S(r) \eqdef \{i \in [n] ~ | ~ \frac{v_i}{c_i} \geq r\}$
        
        \item While $r B < v(S(r)) - \max \{v_i ~ | ~ i \in S(r)\}$ do:
        \begin{enumerate}[itemindent=-0.5em]
            \item Continuously increase ratio $r$
            
            \item If $\frac{v_k}{c_k} \le r$, then discard\footnote{If there are multiple such items, we discard them {\em one by one} in lexicographical order, and stop discarding items once the {\tt While-Loop} meets the {\tt Stop-Condition}.} item $k$: $S(r)\gets S(r)\setminus \{k\}$
        \end{enumerate}
        
        \item Return pair $(r, S(r))$
    \end{enumerate}
    }
    \end{minipage}}
    \caption{The first stage, {\sc Pruning-Mechanism}, shared by all of our mechanisms.}
    \label{fig:pruning}
\end{figure}

\noindent
Noticeably, the set $S(r)$ is always nonempty, since the {\tt Stop-Condition} of the {\tt While-Loop} is violated when $S(r)$ contains only one item.

{\sc Pruning-Mechanism} possesses a remarkable {\em composability property}: the combination of it with any truthful follow-up mechanism $\mathcal{M}$ running on the remaining items $i \in S(r)$ is still a truthful mechanism. More concretely, the composition mechanism $\overline{\mathcal{M}} = (\overline{\allocs}, \overline{\prices})$ of {\sc Pruning-Mechanism} with a follow-up mechanism $\mathcal{M} = (\allocs,\prices)$ works as follows:

\begin{figure}[h]
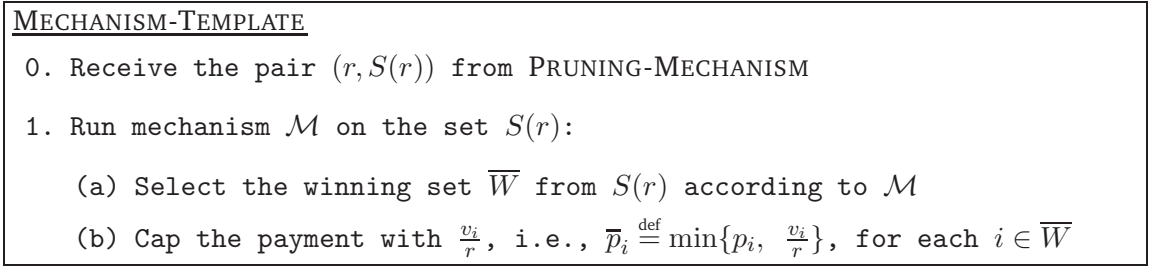

    \centering
    \fbox{\begin{minipage}{0.93\textwidth}
    \underline{\sc Mechanism-Template}
    \vspace{-.5em}
    {\tt \begin{enumerate}[itemindent=-0.5em]
        \setcounter{enumi}{-1}
        \item Receive the pair $(r, S(r))$ from {\sc Pruning-Mechanism}
        
        \item Run mechanism $\mathcal{M}$ on the set $S(r)$:
        \begin{enumerate}[label = (\alph*), itemindent=-0.5em]
            \item Select the winning set $\overline{W}$ from $S(r)$ according to $\mathcal{M}$
            
            \item Cap the payment with $\frac{v_i}{r}$, i.e., $\overline{p}_i \eqdef \min \{p_i, ~ \frac{v_i}{r}\}$, for each $i \in \overline{W}$
        \end{enumerate}
    \end{enumerate}}
    \end{minipage}}
    \caption{The template of a composition mechanism $\overline{\mathcal{M}}$.}
    \label{fig:template}
\end{figure}

\begin{lemma}[Composability]
\label{lem:append}
If a follow-up mechanism $\mathcal{M}$ is individually rational, budget-feasible, and truthful, then so is the composition mechanism $\overline{\mathcal{M}}$.
\end{lemma}

\begin{proof}
By {\tt Step~(1b)} of {\sc Pruning-Mechanism}, every item $i \in S(r)$ has a value-per-cost ratio at least $r$, which means $c_i \leq \frac{v_i}{r}$. Thus, capping the payment with $\frac{v_i}{r}$ does not break the individual rationality. The follow-up mechanism $\mathcal{M}$ itself is budget-feasible, and the composition mechanism $\overline{\mathcal{M}}$ can only reduce the payment for a winning item. Given these, we are left to show the truthfulness of $\overline{\mathcal{M}}$.

%We claim that no winning item $i \in S(r)$ can change the output of {\sc Pruning-Mechanism} by manipulating the bidding cost $c'_{i}$, unless this item gets excluded from $S(r)$ because of an exorbitant $c'_{i}$. Indeed, how large this bidding cost $c'_{i}$ is, only affects {\sc Pruning-Mechanism} at the time point that item $i$ is to be excluded from $S(r)$. Thus, if this item is promised to remain in the output set $i \in S(r)$, the bidding cost $c'_{i}$ cannot change the output.

We claim that no winning item $i \in S(r)$ may change the output of {\sc Pruning-Mechanism} by manipulating its bid to $c'_{i}$, unless this item gets excluded from $S(r)$ because of a too high bid $c'_{i}$. Indeed, suppose item $i$ is still winning with the bid $c'_{i}$, then item $i$ was never removed from the set $S(r)$, i.e., $\frac{v_i}{c'_i} \ge r$ at all times in the {\tt While-Loop} of the {\sc Pruning-Mechanism}. Given that item $i$ stays in the set $S(r)$, the {\tt Stop-Condition} of the {\tt While-Loop} and the order in which we discard other item do not depend on the exact bid $c'_{i}$ of item $i$.

Since the follow-up mechanism $\mathcal{M}$ has a {\em monotone} allocation rule, so does the composition mechanism $\overline{\mathcal{M}}$. Regarding a losing item $i \notin \overline{W}$ (i.e., item $i$ loses in $\overline{\mathcal{M}}$ when it bids truthfully), reporting a higher bid $c'_{i} > c_{i}$ does not help this item to pass the {\sc Pruning-Mechanism} stage. As we discussed above, suppose that item $i$ passes the {\sc Pruning-Mechanism} stage by bidding $c'_{i} > c_{i}$, namely $i \in S'(r')$, then the two outcomes of {\sc Pruning-Mechanism} under the two bids $c'_{i}$ and $c_{i}$ must be the same, namely $(r', S'(r')) = (r, S(r))$. In other words, when item $i$ reports the true cost $c_{i}$, it passes the {\sc Pruning-Mechanism} stage as well, but then loses in the follow-up mechanism $\mathcal{M}$. Given that the follow-up mechanism $\mathcal{M}$ is truthful and runs on the same pair $(r', S'(r')) = (r, S(r))$ in both scenarios, item $i$ will lose again in the follow-up mechanism $\mathcal{M}$, when it reports the higher bid $c'_{i} > c_{i}$.

The payment $\overline{p}_i = \min \{p_i, ~ \frac{v_i}{r}\}$ of the composition mechanism $\overline{\mathcal{M}}$ is exactly the {\em threshold bid} for an item $i \in \overline{W}$ to keep winning: (1)~passing the {\sc Pruning-Mechanism} stage requires a bid of at least $\frac{v_i}{r}$; and (2)~winning in the follow-up mechanism $\mathcal{M}$ (after passing the {\sc Pruning-Mechanism} stage) requires a bid of at least $p_i$.

In addition, a winning item $i \in \overline{W}$ cannot improve its utility by reporting a lower bid $c'_{i} < c_{i}$. As mentioned, when this winning item bids a lower $c'_{i} < c_{i}$, the {\sc Pruning-Mechanism} returns the same pair $(r', S'(r')) = (r, S(r))$. Since the follow-up mechanism $\mathcal{M} = (\allocs,\prices)$ is truthful (i.e.\ a monotone allocation rule and a threshold-based payment rule), item $i$ gets the same payment $p'_{i} = p_{i}$ under either bid $c'_{i}$ or $c_{i}$. The composition mechanism $\overline{\mathcal{M}}$ thus has the same payment $\overline{p}'_{i} = \min \{p'_{i}, \frac{v_{i}}{r'}\} = \min \{p_{i}, \frac{v_{i}}{r}\} = \overline{p}_{i}$ in both scenarios.

This completes the proof of Lemma~\ref{lem:append}.
\end{proof}

We show now several useful properties of the output $(r, S(r))$ of {\sc Pruning-Mechanism}.
%in Lemma~\ref{lem:pruning:approx} below.
\begin{lemma}[Pruning Mechanism]
\label{lem:pruning:approx}
Let $i^* \in \arg\max \{v_i ~ | ~ i \in S(r)\}$ denote the highest-value item or one of the highest-value items,\footnote{When there are multiple  highest-value items, we break ties lexicographically.} and let $T \eqdef S(r) \setminus \{i^*\}$. Then the following hold:
\begin{enumerate}[label=(\alph*)., font = {\em\bfseries}]
\item $c_i \leq \frac{v_i}{r} \leq B$ for each item $i \in S(r)$.
\item $v(T) \leq r B < v(S(r))$.
\item $\fopt \leq v(S(r)) + r \cdot (B - c(S(r)) < 2 \cdot v(S(r))$.
\end{enumerate}
\end{lemma}

\begin{proof}
{\bf Property~(a).}
The first inequality follows from {\tt Step~(1b)} of {\sc Pruning-Mechanism}; the second inequality holds, since the ratio $r$ is initialized to be $\frac{1}{B} \cdot \max\{v_i ~ | ~ i \in [n]\}$, and keeps increasing during the {\tt While-Loop}.

{\bf Property~(b).}
We observe that the first inequality is a reformulation of the {\tt Stop-Condition} of the {\tt While-Loop}. To prove the second inequality, we note that there are two possibilities that can lead to the termination of the {\tt While-Loop}, and $r B < v(S(r))$ holds in both cases.
\begin{itemize}%[itemindent=-1.2em]
\item {\tt [Increase of ratio $r$].} Continuous increase of $r$ implies $r B = v(T) < v(S(r))$.
\item {\tt [Discard of an item $k$].} Value-per-cost ratio $r$ is fixed before and after the discard. Before the discard, in that {\tt Stop-Condition} has not been invoked,
    \[
    r B < v(S(r)) + v_k - \max \{v_{i^*}, ~ v_k\} \leq v(S(r)).
    \]
\end{itemize}

{\bf Property~(c).}
The second inequality follows from Property~(b). We show the first inequality based on case analysis. Let $\allocs = (\alloci)_{i = 1}^{n}$ denote the solution to the {\sf Fractional Knapsack} problem. We have either $S(r) \subseteq \{i \in [n] ~ | ~ \alloci = 1\}$ or $S(r) \supseteq \{i \in [n] ~ | ~ \alloci > 0\}$. This claim holds since: (1)~{\sc Pruning-Mechanism} discards the items in increasing order of the value-per-cost ratios; but (2)~the greedy algorithm takes the items in decreasing order of the value-per-cost ratios; and (3)~in both processes, we break ties lexicographically.
\begin{itemize}
\item {\tt [When $S(r) \subseteq \{i \in [n] ~ | ~ \alloci = 1\}$].}
    We notice that $c(S(r)) \leq \sum_{i \in [n]} x_{i} \cdot c_{i} \le B$. Namely, regarding the {\sf Fractional Knapsack} optimum, the total cost $\sum_{i \in [n]} x_{i} \cdot c_{i}$ is at least the cost on the items in $S(r)$, and is at most the budget $B$. In addition, every item $i \notin S(r)$ has a value-per-cost ratio $\frac{\vali}{\costi}\le r$. Consequently, the total value of the items beyond set $S(r)$ is $\sum_{i \notin S(r)} \alloci \cdot \vali \le r \cdot \sum_{i \notin S(r)} \alloci \cdot c_i \leq r \cdot (B - c(S(r)))$.
    
    \item {\tt [When $S(r) \supseteq \{i \in [n] ~ | ~ \alloci> 0\}$].}
    We have $\sum_{i \in [n]} x_{i} \cdot c_{i} \leq c(S(r))$ and $\sum_{i \in [n]} x_{i} \cdot c_{i} \leq B$, and every item $i \in S(r)$ has a value-per-cost ratio $\frac{\vali}{\costi} \ge r$. As a result, $v(S(r)) - \fopt = \sum_{i \in S(r)} (1 - \alloci) \cdot \vali \ge r \cdot \sum_{i \in S(r)} (1 - \alloci) \cdot c_i \geq r \cdot (c(S(r)) - B)$.
\end{itemize}

This completes the proof of properties~(a), (b), and (c).
\end{proof}

\noindent
{\bf Mechanisms in the Second Stage.}
Given Lemma~\ref{lem:append}, {\sc Pruning-Mechanism} can be composed with any follow-up truthful mechanism. Actually, we focus on the class of {\em posted-price} mechanisms.\footnote{To obtain our $3$-approximation deterministic mechanism in Section~\ref{sec:det}, we actually use an adaptive posted-price scheme. Namely, the take-it-or-leave price offered to a specific item $i \in S(r)$ can change, depending on whether the items that have already made decisions accepted or rejected their posted-price offers.} Such a mechanism is determined by a set of prices $(B_i)_{i \in S(r)}$ subject to the budget constraint $\sum_{i\in S(r)}B_i\le B$, and naturally meets the individual rationality, the budget feasibility, and the truthfulness.\footnote{In the case of a randomized mechanism, any realization is given by a particular set of budget-feasible posted prices $(B_i)_{i \in S(r)}$, i.e., a truthful deterministic mechanism. Thus, this randomized mechanism is universally truthful.}

To illustrate how to analyze the approximability of a two-stage posted-price mechanism, and as a warm-up exercise, below we discuss two simple mechanisms.

\vspace{.1in}
\noindent
{\bf Warm-Up.}
Our first mechanism (see Figure~\ref{fig:first_warm_up}) chooses the higher-value subset between $\{i^*\}$ and $T$ as the winning set $W$, where $i^* \in \arg\max \{v_i ~ | ~ i \in S(r)\}$ is the highest-value item and $T = S(r) \setminus \{i^*\}$ (see Lemma~\ref{lem:pruning:approx}), by offering price $\frac{\vali}{r}$ to each $i \in \{i^*\}$ or to each $i \in T$. Hence, we deduce from Lemma~\ref{lem:pruning:approx}~(c) that $\fopt \leq 2 \cdot v(S(r)) \leq 4 \cdot \max \{v_{i^*}, ~ v(T)\} = 4 \cdot \alg$.

\begin{figure}[htbp]
    \centering
    \fbox{\begin{minipage}{0.82\textwidth}
    \underline{\sc First-Warm-Up-Mechanism}
    \vspace{-.5em}
    {\tt \begin{enumerate}[itemindent=-0.5em]
        \setcounter{enumi}{-1}
        \item Receive the pair $(r, S(r))$ from {\sc Pruning-Mechanism}
        
        \item If $v_{i^*} \geq v(T)$,\footnote{Item $i^*$ will accept the offer $\frac{v_{i^*}}{r}$, by Lemma~\ref{lem:pruning:approx} (a) that $c_{i^*} \leq \frac{v_{i^*}}{r}$.} get item $i^*$ by offering price $\frac{v_{i^*}}{r}$
        
        \item Else,\footnote{Each item $i \in T$ will accept the offer $\frac{v_i}{r}$, by Lemma~\ref{lem:pruning:approx} (a) that $c_i \le \frac{v_i}{r}$.} get items $T$ by offering price $\frac{v_i}{r}$ to each item $i \in T$
    \end{enumerate}}
    \end{minipage}}
    \caption{A $4$-approximation deterministic budget-feasible mechanism.}
    \label{fig:first_warm_up}
\end{figure}

Our second posted-price mechanism (see Figure~\ref{fig:second_warm_up}) recovers the best known result of $(2+ \sqrt{2})$ by Chen et al.~\cite{CGL11}. This statement is formalized as the following theorem.

\begin{figure}[htbp]
    \centering
    \fbox{\begin{minipage}{0.77\textwidth}
    \underline{\sc Second-Warm-Up--Mechanism}
    \vspace{-.5em}
    {\tt \begin{enumerate}[itemindent=-0.5em]
        \setcounter{enumi}{-1}
        \item Receive the pair $(r, S(r))$ from {\sc Pruning-Mechanism}
        
        \item If $v_{i^*} \geq \sqrt{2} \cdot v(T)$, get item $i^*$ by offering price $\frac{v_{i^*}}{r}$
        
        \item Else,
        \begin{enumerate}[itemindent=-0.5em]
            \item Get items $T$ by offering price $\frac{v_i}{r}$ to each item $i \in T$
            
            \item Offer price\footnote{Notice from Lemma~\ref{lem:pruning:approx}~(b) that $0 \leq B - \frac{v(T)}{r} < \frac{v_{i^*}}{r}$.} $B - \frac{v(T)}{r}$ to item $i^*$
        \end{enumerate}
    \end{enumerate}}
    \end{minipage}}
    \caption{A new $(2 + \sqrt{2})$-approximation deterministic budget-feasible mechanism.}
    \label{fig:second_warm_up}
\end{figure}

\begin{theorem}
\label{thm:warmup}
{\sc Second-Warm-Up--Mechanism} is a $(2+ \sqrt{2})$-approximation mechanism (individually rational, budget-feasible, and truthful) against the {\sf Fractional Knapsack} benchmark.
\end{theorem}

\begin{proof}
We only show the approximability via case analysis; the other properties are obvious.
\begin{itemize}
    \item {\tt [Case~1 that $v_{i^*} \geq \sqrt{2} \cdot v(T)$].}
    The highest-value item $i^*$ is the only winner, and thus the outcome value $\alg = v_{i^*}$. Then according to Lemma~\ref{lem:pruning:approx}~(c), we have
    \[
        \fopt
        \leq 2 \cdot v(S(r))
        = 2 \cdot (v_{i^*} + v(T))
        \leq (2 + \sqrt{2}) \cdot v_{i^*}
        = (2 + \sqrt{2}) \cdot \alg.
    \]
    
    \item {\tt [Case~2 that $v_{i^*} < \sqrt{2} \cdot v(T)$].}
    There are two possibilities. First, when $c_{i^*} \leq B - \frac{v(T)}{r}$, all items $i \in S(r)$ together form the winning set $W$, i.e., $\alg = v(S(r))$. Due to Lemma~\ref{lem:pruning:approx}~(c), $\fopt \leq 2 \cdot v(S(r)) = 2 \cdot \alg$. Second, when $c_{i^*} > B - \frac{v(T)}{r}$, only the items $i \in T$ are chosen as the winners, i.e., $\alg = v(T)$. Consequently,
    \begin{align*}
        \fopt
        & \leq v(S(r)) + r \cdot (B - c(S(r))) && \text{(Lemma~\ref{lem:pruning:approx} (c))} \\
        & \leq v(S(r)) + v(T)
        && \text{(as $c(S(r)) \geq c_{i^*} > B - \tfrac{v(T)}{r}$)} \\
        & = v(i^*) + 2 \cdot v(T)
        && \text{(as $v(S(r)) = v(i^*) + v(T)$)} \\
        & < (2+ \sqrt{2}) \cdot \alg.
        && \text{(as $v_{i^*} < \sqrt{2} \cdot v(T) = \sqrt{2} \cdot \alg$)}.
    \end{align*}
\end{itemize}
This completes the proof of Theorem~\ref{thm:warmup}
\end{proof}

We emphasizes that our {\sc Second-Warm-Up--Mechanism} achieves a $2$-approximation, when $v_{i^*} < \sqrt{2} \cdot v(T)$ and $c_{i^*} \le B - \frac{v(T)}{r}$. One might ask a natural question: is it possible to achieve a better trade-off between this $2$-approximation case and the $(2+ \sqrt{2})$-approximation cases? In the next section, we will confirm this guess by presenting a slightly more complicated {\em adaptive} posted-price scheme, resulting in a $3$-approximation deterministic mechanism.

%% file: deterministic.tex
The warm-up mechanisms have merely a few possible outcomes, and do not adapt to the decisions of the items:  either the highest-value item $i^*$, or the remaining items $T$, or rarely both of item $i^*$ and items $T$ win; all the posted prices $(B_i)_{i \in S(r)}$ are almost equal to the maximum possible values $(\frac{\vali}{r})_{i \in S(r)}$. Such rigid structure hinders both warm-up mechanisms from achieving better performance guarantees than a $(2+ \sqrt{2})$-approximation.

Now we give a mechanism (called {\sc Deterministic-Mechanism}) that achieves a better approximation. This mechanism (first stage) gets the pair $(r,S(r))$ via the {\sc Pruning-Mechanism} given in Section~\ref{sec:pruning}, and then (second stage) applies an adaptive posted-price scheme.

\begin{figure}[htbp]
    \centering
    \fbox{\begin{minipage}{0.94\textwidth}
    \underline{\sc Deterministic-Mechanism}
    \vspace{-.5em}
    {\tt \begin{enumerate}[itemindent=-0.5em]
        \setcounter{enumi}{-1}
        \item Receive the pair $(r, S(r))$ from {\sc Pruning-Mechanism}
        
        \item If $v_{i^*} \leq \frac{1}{2} \cdot v(T)$, get items $T$ by offering price $\frac{v_{i}}{r}$ to each item $i\in T$
        
        \item Else if $v_{i^*} \geq 2 \cdot v(T)$, get item $i^*$ by offering price $\frac{v_{i^*}}{r}$ to $i^*$
        
        \item Else, i.e., when $\frac{1}{2} \cdot v(T) < v_{i^*} < 2 \cdot v(T)$:
        \begin{enumerate}[itemindent=-0.5em]
            \item Offer price $B_{i^*} \eqdef \min \{\frac{v_{i^*}}{r}, ~ \frac{2 \cdot v_{i^*} - v(T)}{v(S(r))} \cdot B\}$ to item $i^*$
            
            \item If $c_{i^*} \leq B_{i^*}$,\footnote{If $c_{i^*} \leq B_{i^*}$, item $i^*$ will accept offer $B_{i^*}$. Otherwise, $c_{i^*} > B_{i^*}$ and item $i^*$ will reject offer $B_{i^*}$, and then each item $i \in T$ will accept offer $B_i$.} offer $B_i \eqdef \min \{\frac{v_{i}}{r}, ~ \frac{v_i}{v(T)} \cdot (B - B_{i^*})\}$ to each item $i \in T$
            
            \item Else, get items $T$ by offering price $\frac{v_{i}}{r}$ to each item $i\in T$
        \end{enumerate}
    \end{enumerate}}
    \end{minipage}}
    \caption{The $3$-approximation deterministic budget-feasible mechanism.}
    \label{fig:deterministic}
\end{figure}

% We also call by the same name {\sc Deterministic-Mechanism} the composition of two mechanisms: {\sc Pruning-Mechanism} with {\sc Deterministic-Mechanism}.

\begin{theorem}
\label{thm:dtm}
{\sc Deterministic-Mechanism} is a $3$-approximation mechanism (individually rational, budget-feasible, and truthful) against the {\sf Fractional Knapsack} benchmark.
\end{theorem}

\begin{proof}
The individual rationality and the truthfulness are easy to see, regarding the pricing nature of {\sc Deterministic-mechanism}, Lemma~\ref{lem:append}, and Lemma~\ref{lem:pruning:approx}~(a). To show the budget feasibility, we consider either {\tt Case~(3b)} or {\tt Case~(3c)} in the mechanism:
\begin{itemize}
\item {\tt [Case~(3b)].} $\sum_{i \in W} B_i \leq B_{i^*} + \sum_{i \in T} \frac{v_i}{v(T)} \cdot (B - B_{i^*}) = B$.
\item {\tt [Case~(3c)].} Since $W = T$, we know from Lemma~\ref{lem:pruning:approx}~(b) that $\sum_{i \in W} \frac{v_i}{r} = \frac{v(T)}{r} \leq B$.
\end{itemize}

We now show the approximation guarantee. Both of {\tt Case~(1)} and {\tt Case~(2)}, where either $v_{i^*} \leq \frac{1}{2} \cdot v(T)$ or $v_{i^*} \geq 2 \cdot v(T)$, are easy to analyze. Since $\alg = \max\{v_{i^*}, ~ v(T)\}$ in either case,
\begin{align*}
    \fopt
    < 2 \cdot v(S(r))
    = 2 \cdot (v_{i^*} + v(T))
    \le 3 \cdot \max\{v_{i^*}, ~ v(T)\}
    = 3 \cdot \alg,
\end{align*}
where the first step applies Lemma~\ref{lem:pruning:approx}~(c), and the third step holds since we have $2 \cdot v_{i^*} \le v(T)$ or $v_{i^*} \geq 2 \cdot v(T)$ in both cases.

From now on, we safely assume $\frac{1}{2} \cdot v(T) < v_{i^*} < 2 \cdot v(T)$. Conditioned on either $c_{i^*} \leq B_{i^*}$ or $c_{i^*} > B_{i^*}$, we are only left to deal with {\tt Case~(3b)} and {\tt Case~(3c)}.

\vspace{.1in}
\noindent
{\tt [Case~(3b) that $c_{i^*} \leq B_{i^*}$].}
We denote by $U \eqdef \left\{i \in T ~ | ~ c_i \leq B_i\right\}$ the set of winners in $T$, so the outcome value $\alg = v_{i^*} + v(U)$. Of course, a losing item $i \in (T \setminus U)$ rejects the offered price $B_i = \min \{\frac{v_i}{r}, ~ \frac{v_i}{v(T)} \cdot (B - B_{i^*})\}$ (by definition), since it has a too large cost $c_{i} > B_{i}$. But this losing item was not discarded during {\sc Pruning-Mechanism}, so it has a high enough value-per-cost ratio $\frac{v_{i}}{c_{i}} \geq r$ (see Lemma~\ref{lem:pruning:approx}~(a)) and thus a cost $c_{i} \leq \frac{v_{i}}{r}$. For these reasons, the price offered to this losing item is exactly $B_{i} = \frac{v_i}{v(T)} \cdot (B - B_{i^*})$. We deduce that
%We have $c_i > B_i$ for each item $i \in (T \setminus U)$. Note that $B_i < \frac{\vali}{r}$ for each item $i \in (T \setminus U)$, since otherwise item $i$ would accept price $B_i$ and would actually get in set $U$. Thus, $c_i > B_i = \frac{v_i}{v(T)} \cdot (B - B_{i^*})$ for each item $i \in (T \setminus U)$ and
\begin{equation}
    \label{eq:dtm:cost}
    c(S(r)) \geq \sum_{i \in (T \setminus U)} c_i > \sum_{i \in (T \setminus U)} B_i = \tfrac{v(T \setminus U)}{v(T)} \cdot (B - B_{i^*}).
\end{equation}

By Lemma~\ref{lem:pruning:approx}~(c), $\fopt \leq v(S(r)) + r \cdot (B - c(S(r)))$. We plug inequality~\eqref{eq:dtm:cost} into it and get
\begin{align*}
    \fopt
    ~ \overset{\eqref{eq:dtm:cost}}{<} ~ & v(S(r)) + r \cdot (\tfrac{v(U)}{v(T)} \cdot B + \tfrac{v(T \setminus U)}{v(T)} \cdot B_{i^*}) \\
    ~ < ~ & v(S(r)) \cdot (1 + \tfrac{v(U)}{v(T)} + \tfrac{v(T \setminus U)}{v(T)} \cdot \tfrac{B_{i^*}}{B})
    && \text{(Lemma~\ref{lem:pruning:approx}~(b): $r B < v(S(r))$)} \\
    ~ \le ~ & v(S(r)) \cdot (1 + \tfrac{v(U)}{v(T)}) + \tfrac{v(T \setminus U)}{v(T)} \cdot (2 \cdot v_{i^*} - v(T))
    && \text{(as $B_{i^*} \le \tfrac{2 \cdot v_{i^*} - v(T)}{v(S(r))}\cdot B$)} \\
    ~ = ~ & 3 \cdot v_{i^*} + v(U) \cdot (2 - \tfrac{v_{i^*}}{v(T)})
    && \text{(as $v(S(r)) = v_{i^*} + v(T)$)} \\
    ~ \leq ~ & 3 \cdot v_{i^*} + 3 \cdot v(U)
    ~ = ~ 3 \cdot \alg.
\end{align*}

\noindent
{\tt [Case~(3c) that $c_{i^*} > B_{i^*}$].}
According to Lemma~\ref{lem:pruning:approx}~(a), $c_{i^*} \le \frac{v_{i^*}}{r}$, and $c_i \le \frac{v_i}{r}$ for any $i \in T$. Since $B_{i^*} < c_{i^*} \leq \frac{v_{i^*}}{r}$, we have $B_{i^*} = \min \{\frac{v_{i^*}}{r}, ~ \frac{2 \cdot v_{i^*} - v(T)}{v(S(r))} \cdot B\} = \frac{2 \cdot v_{i^*} - v(T)}{v(S(r))} \cdot B$.

In this case, the highest-value item $i^*$ rejects its offer, but all the remaining items $i \in T$ accept their offers. Thus, the winning set is $W = T$, and the outcome value is $\alg = v(T)$. We then deduce that
\begin{align*}
    \fopt
    ~ \leq ~ & v(S(r)) + r \cdot (B - c(S(r)))
    && \text{(Lemma~\ref{lem:pruning:approx}~(c))} \\
    ~ \leq ~ & v(S(r)) + r \cdot (B - B_{i^*}) && \text{(as $c(S(r)) \geq c_{i^*} > B_{i^*}$)} \\
    ~ \leq ~ & v(S(r)) \cdot (2 - \tfrac{B_{i^*}}{B}) && \text{(Lemma~\ref{lem:pruning:approx}~(b): $r B < v(S(r))$)} \\
    ~ = ~ & 3 \cdot v(T)
    ~ = ~ 3 \cdot \alg.
    && \text{(as $B_{i^*} = \tfrac{2 \cdot v_{i^*} - v(T)}{v(S(r))} \cdot B$)}
\end{align*}

To conclude, we have $3 \cdot \alg \ge \fopt$ in all cases, which completes the proof of Theorem~\ref{thm:dtm}.
\end{proof}

\subsection{Matching Lower Bound}
\label{subsec:dtm:lower}
Against the {\sf Fractional Knapsack} benchmark, our {\sc Deterministic-Mechanism} turns out to have the best possible approximation ratio among all deterministic mechanisms. To see so, we now construct a matching lower-bound instance, which is similar to~\cite[Proposition~5.2]{S10}.
\begin{theorem}
\label{thm:lower:dtm}
No deterministic mechanism (truthful, individually rational and budget-feasible) has an approximation ratio less than $3$ against the {\sf Fractional Knapsack} benchmark, even if there are only three items.
\end{theorem}
\begin{proof}
For the sake of contradiction, assume that there is a $(3 - \varepsilon)$-approximation deterministic mechanism, for some constant $\varepsilon > 0$. Consider the following two scenarios with three items having values $v_1 = v_2 = v_3 = 1$. Let $c^* \eqdef \frac{B}{2 - \varepsilon/2}$; notice that $2 c^* > B$.
\begin{itemize}
\item {\tt [With costs $(c^*, c^*, c^*)$].} Due to the individual rationality, each winning item gains a payment of at least $c^*$. To guarantee the promised approximation ratio of $(3 - \varepsilon)$ under budget feasibility, there is exactly one winning item. W.l.o.g., we assume that the winner is the first item.
\item {\tt [With costs $(0, c^*, c^*)$].} By the truthfulness, item $1$ wins once again, getting the same payment of at least $c^*$. As a result, the budget left is at most $(B - c^*) < c^*$. Regarding the budget feasibility and individual rationality, neither item $2$ nor item $3$ can win.
\end{itemize}
In the later scenario, the mechanism generates value $\alg = 1$, yet the {\sf Fractional Knapsack} benchmark achieves value $\fopt = 1 + \frac{B}{c^*} = 3 - \frac{\varepsilon}{2} > 3 - \varepsilon$. This contradicts our assumption that the mechanism is $(3 - \varepsilon)$-approximation, concluding the proof of the theorem.
\end{proof}

%% file: randomized.tex
We now present the main result of our work, a randomized mechanism (called {\sc Randomized-Mechanism}) that achieves a $2$-approximation to the {\sf Fractional Knapsack} benchmark. Regarding the matching lower bound by Chen et al.~\cite[Theorem~4.2]{CGL11} against the weaker {\sf Knapsack} benchmark, this approximation guarantee is tight for both benchmarks. Our mechanism (first stage) gets the pair $(r,S(r))$ from the {\sc Pruning-Mechanism} given in Section~\ref{sec:pruning}, and then (second stage) applies a randomized non-adaptive posted-price scheme.

\begin{figure}[htbp]
    \centering
    \fbox{\begin{minipage}{0.89\textwidth}
    \underline{\sc Randomized-Mechanism}
    \vspace{-.5em}
    {\tt \begin{enumerate}[itemindent=-0.5em]
        \setcounter{enumi}{-1}
        \item Receive the pair $(r, S(r))$ from {\sc Pruning-Mechanism}
        
        \item Let $q \eqdef \frac{1}{2} \cdot \frac{v(S(r)) - r B}{\min\{v_{i^*}, ~ v(T)\}}$
        
        \item If $v_{i^*} \leq v(T)$, let $q_{i^*} \eqdef (\frac{1}{2} - q)$ and $q_T \eqdef \frac{1}{2}$
        
        \item Else, let $q_{i^*} \eqdef \frac{1}{2}$ and $q_T \eqdef (\frac{1}{2} - q)$
        
        \item Offer price\footnote{For every item $i \in S(r)$, price $B_i$ is well defined in range $[0, ~ \frac{v_i}{r}] \subseteq [0, ~ B]$, by Lemma~\ref{lem:pruning:approx}~(a).} $B_{i^*}$ to item $i^*$, where $B_{i^*}$ is defined as follows:
        \begin{enumerate}[itemindent=-0.5em]
            \item With probability $q_{i^*}$, let $B_{i^*} \eqdef \frac{v_{i^*}}{r}$
            
            \item With probability $q_{T}$, let $B_{i^*} \eqdef B - \frac{v(T)}{r}$
            
            \item With probability $q$, draw $B_{i^*} \sim {\tt Uniform} [B - \frac{v(T)}{r}, ~ \frac{v_{i^*}}{r}]$
        \end{enumerate}
        
        \item Offer price $B_i \eqdef \frac{v_i}{v(T)} \cdot (B - B_{i^*})$ to each item $i \in T$
    \end{enumerate}}
    \end{minipage}}
    \caption{The $2$-approximation randomized budget-feasible mechanism.}
    \label{fig:randomized}
\end{figure}

We first verify that all quantities in {\sc Randomized-Mechanism} are well defined.

\begin{lemma}
\label{lem:random:p}
$0 \leq q = \frac{1}{2} \cdot \frac{v(S(r)) - r B}{\min\{v_{i^*}, ~ v(T)\}} \leq \frac{1}{2}$ and $0 \leq B - \frac{v(T)}{r} < \frac{v_{i^*}}{r}$.
\end{lemma}

\begin{proof}
The first inequality is due to Lemma~\ref{lem:pruning:approx}~(b) that $v(S(r)) > r B$. Lemma~\ref{lem:pruning:approx} further implies $r B \geq v_{i^*}$ and $r B \geq v(T)$, i.e., $r B \geq \max \{v_{i^*}, ~ v(T)\}$. Now, the second inequality in Lemma~\ref{lem:random:p} follows, as $q = \frac{1}{2} \cdot \frac{v_{i^*} + v(T) - r B}{\min\{v_{i^*}, ~ v(T)\}} \leq \frac{1}{2} \cdot \frac{v_{i^*} + v(T) - \max\{v_{i^*}, ~ v(T)\}}{\min\{v_{i^*}, ~ v(T)\}} = \frac{1}{2}$. Finally, rearranging $v(T) \leq r B < v(S(r)) = v_{i^*} + v(T)$ leads to the last two inequalities.
\end{proof}

Similar to {\sc Deterministic-Mechanism} in Section~\ref{sec:det}, we also slightly abuse notations and also refer to {\sc Randomized-Mechanism} as the composition of two mechanisms: {\sc Pruning-Mechanism} with {\sc Randomized-Mechanism}.

\begin{theorem}
\label{thm:rdm}
{\sc Randomized-Mechanism} is a $2$-approximation mechanism (individually rational, budget-feasible, and universally truthful) against the {\sf Fractional Knapsack} benchmark.
\end{theorem}

\begin{proof}
Since {\sc Randomized-Mechanism} is a posted-price scheme, it is individually rational. Since each random realization of the prices $(B_i)_{i \in S(r)}$ is budget-feasible, i.e., $\sum_{i \in S(r)} B_i = B$ by construction, the mechanism is also budget-feasible. Note that (1)~all random choices in {\sc Randomized-Mechanism}, i.e., the prices $(B_{i})_{i \in S(r)}$, can be made before execution of the mechanism; and (2)~for each such choice, the resulting posted-price mechanism is obviously truthful. Due to Lemma~\ref{lem:append}, all desired properties extend to the composition mechanism, hence being individually rational, budget-feasible, and universally truthful.

In the rest of the proof, we show that {\sc Randomized-Mechanism} is a $2$-approximation to $\fopt$. Let $(\alloci)_{i \in S(r)}$ denote the allocation probabilities, then the mechanism generates an expected value of $\alg = \sum_{i \in S(r)} v_i \cdot \alloci$. In order to prove the approximation guarantee, we need the following equation~\eqref{eq:rdm:identity}, inequality~\eqref{eq:rdm:i^*}, and inequality~\eqref{eq:rdm:T}, which will be proved later.
\begin{align}
    r B
    ~ = ~ & 2q_{i^*} \cdot v_{i^*} + 2q_T \cdot v(T),
    \label{eq:rdm:identity} \\
    v_{i^*} \cdot \alloc_{i^*}
    ~ \geq ~ & q_{i^*} \cdot v_{i^*} + \tfrac{1}{2} \cdot (v_{i^*} - r \cdot c_{i^*}),
    \label{eq:rdm:i^*} \\
    \vali \cdot \alloci
    ~ \geq ~ & q_T \cdot \vali + \tfrac{1}{2} \cdot (\vali - r \cdot \costi), \quad\quad \forall i \in T.
    \label{eq:rdm:T}
\end{align}
Indeed, these mathematical facts together with Lemma~\ref{lem:pruning:approx}~(c) imply that $2 \cdot \alg \ge \fopt$.
\begin{align*}
    \fopt
    \,\,\,\leq\,\,\, & v(S(r)) + r \cdot (B - c(S(r))) \\
    \,\,\overset{\eqref{eq:rdm:identity}}{=}\,\, & (v_{i^*} + v(T)) + 2 \cdot (q_{i^*} \cdot v_{i^*} + q_T \cdot v(T)) - r \cdot (c_{i^*} + c(T)) \\
    \overset{(\ref{eq:rdm:i^*},\ref{eq:rdm:T})}{\leq} & 2v_{i^*} \cdot \alloc_{i^*} + 2 \cdot \sum_{i \in T} v_i \cdot \alloci \\
    \,\,\,=\,\,\, & 2 \cdot \alg.
\end{align*}
Now, we are only left to prove equation~\eqref{eq:rdm:identity}, inequality~\eqref{eq:rdm:i^*} and inequality~\eqref{eq:rdm:T}.

\vspace{.1in}
\noindent
{\tt [Equation~\eqref{eq:rdm:identity}].} By the definitions of $q_{i^*}$ and $q_T$, in either case of {\tt Step~(2)} or {\tt Step~(3)},
\begin{align*}
    q_{i^*} \cdot v_{i^*} + q_T \cdot v(T)
    ~ = ~ & \tfrac{1}{2} \cdot (v_{i^*} + v(T)) - q \cdot \min\{v_{i^*}, ~ v(T)\} \\
    ~ = ~ & \tfrac{1}{2} \cdot (v_{i^*} + v(T)) - \tfrac{1}{2} \cdot (v(S(r)) - r B)
    && \text{(definition of $q$)} \\
    ~ = ~ & \tfrac{1}{2} \cdot r B.
\end{align*}

\noindent
{\tt [Inequality~\eqref{eq:rdm:i^*}].} It is equivalent to showing that $\Prx{B_{i^*} \geq c_{i^*}} = \alloc_{i^*} \ge q_{i^*} + \frac{v_{i^*} - r \cdot c_{i^*}}{2v_{i^*}}$.
\begin{itemize}
    \item {\tt [When $c_{i^*} \leq B - \frac{v(T)}{r}$].} Item $i^*$ always accepts price $B_{i^*}$, i.e., $\Prx{B_{i^*} \geq c_{i^*}} = 1$, which gives us the desired bound of $1\ge q_{i^*} + \frac{v_{i^*} - r \cdot c_{i^*}}{2v_{i^*}}$, because $q_{i^*} \leq \frac{1}{2}$.

    \item {\tt [When $c_{i^*} > B - \frac{v(T)}{r}$].} Due to Lemma~\ref{lem:pruning:approx}~(a), $\frac{v_{i^*}}{r} \geq c_{i^*}$. We consider the random events in ${\tt Step~(4a)}$ that $B_{i^*} = \frac{v_{i^*}}{r}$ and in ${\tt Step~(4c)}$ that $B_{i^*} \sim {\tt Uniform}[B - \frac{v(T)}{r}, ~ \frac{v_{i^*}}{r}]$. Since $\Prx{\tt Step~(4a)} = q_{i^*}$ and $\Prx{\tt Step~(4c)} = q$, putting everything together gives
\begin{align*}
    \Prx{B_{i^*} \geq c_{i^*}}
    ~ = ~ & q_{i^*} + q \cdot \frac{v_{i^*} / r - c_{i^*}}{v_{i^*} / r - (B - v(T) / r)}
    && \text{(Lemma~\ref{lem:pruning:approx}~(a): $\tfrac{v_i}{r} \ge c_i$)} \\
    ~ = ~ & q_{i^*} + \frac{1}{2} \cdot \frac{v_{i^*} - r \cdot c_{i^*}}{\min\{v_{i^*}, ~ v(T)\}}
    && \text{(definition of $q$)} \\
    ~ \geq ~ & q_{i^*} + \frac{v_{i^*} - r \cdot c_{i^*}}{2v_{i^*}}.
    && \text{(as $v_{i^*} \ge \min \{v_{i^*}, ~ v(T)\}$ and $\tfrac{v_i}{r} \ge c_i$)}
\end{align*}
\end{itemize}

\noindent
{\tt [Inequality~\eqref{eq:rdm:T}].} The argument is similar to the above. For each item $i \in T$, we claim that $\Prx{B_i \ge c_i}  = \alloci \ge q_T + \frac{v_i - r \cdot c_i}{2v_i}$.
\begin{itemize}
\item {\tt [When $c_i \leq \frac{v_i}{v(T)} \cdot (B - \frac{v_{i^*}}{r})$].} Item $i$ always accepts price $B_i$, i.e., $\Prx{B_i \geq c_i} = 1$, which gives us the desired bound of $1 \ge q_T + \frac{v_i - r \cdot c_i}{2v_i}$, in that $q_{T} \le \frac{1}{2}$.
\item {\tt [When $c_i > \frac{v_i}{v(T)} \cdot (B - \frac{v_{i^*}}{r})$].} By {\tt Step~(5)}, $B_i \geq c_i$ if and only if $B_{i^*} \leq B - v(T) \cdot \frac{c_i}{v_i}$. We consider the random events in ${\tt Step~(4b)}$ that $B_{i^*} = B - \frac{v(T)}{r}$ and in ${\tt Step~(4c)}$ that $B_{i^*} \sim {\tt Uniform}[B - \frac{v(T)}{r}, ~ \frac{v_{i^*}}{r}]$. Because $\Prx{\tt Step~(4b)} = q_{T}$ and $\Prx{\tt Step~(4c)} = q$,
\begin{align*}
    \Prx{B_i \geq c_i}
    ~ = ~ & q_T + q \cdot \frac{(B - v(T) \cdot c_i / v_i) - (B - v(T) / r)}{v_{i^*} / r - (B - v(T) / r)} \\
    ~ = ~ & q_T + q \cdot \frac{v(T)}{v(S(r)) - r B} \cdot \frac{v_i - r \cdot c_i}{v_i}
    && \text{(as $v(S(r)) = v_{i^*} + v(T)$)} \\
    ~ = ~ & q_T + \frac{1}{2} \cdot \frac{v(T)}{\min\{v_{i^*}, ~ v(T)\}} \cdot \frac{v_i - r \cdot c_i}{v_i} && \text{(definition of $q$)} \\
    ~ \ge ~ & q_T + \frac{v_i - r \cdot c_i}{2v_i}, && \text{(Lemma~\ref{lem:pruning:approx}~(a): $\tfrac{v_i}{r} \ge c_i$)}
\end{align*}
\end{itemize}
This completes the proof of Theorem~\ref{thm:rdm}.
\end{proof}

%% file: open.tex
%Conclusions
In this work, we proposed a budget-feasible randomized mechanism with the best possible approximation guarantee for an additive buyer. In addition, our deterministic mechanism still leaves some room for improvement: the best possible approximation guarantee is somewhere between $\big[\sqrt{2} + 1, 3\big]$. However, our instance from Theorem~\ref{thm:lower:dtm} clearly demonstrates that quite a different approach that is specifically tailored to the real \textsf{Knapsack} optimum (rather than the fractional relaxation solution) is needed.

The class of additive valuations is the most basic class of valuations in the research agenda for budget-feasible mechanisms. We hope that our results may lead to new mechanisms and improved analysis for broader valuation classes. Indeed, given the same factor $2$-approximation result of~\cite{BCGL17} for the promise version of the problem for a subadditive buyer, we are even so bold as to conjecture that the true approximation guarantee for a subadditive buyer is still $2$ (leaving all computational considerations aside).

Our composition approach has a lot of resemblance to the pruning ideas from the frugality literature. This demonstrates that ideas and approaches from one area of reverse auction design might be beneficial to another. We believe that there could be more interesting connections between these two complementary agendas.

Finally, our mechanisms use posted prices in the second stage. Besides the practical interest and motivation of posted-price mechanisms in the prior work, our work gives additional support to study this family of mechanisms in budget-feasible framework from a theoretical viewpoint.

%% file: main.bbl
\begin{thebibliography}{CEGP10}

\bibitem[ABM16]{ABM16}
Georgios Amanatidis, Georgios Birmpas, and Evangelos Markakis.
\newblock Coverage, matching, and beyond: New results on budgeted mechanism
  design.
\newblock In {\em Web and Internet Economics - 12th International Conference,
  {WINE} 2016, Montreal, Canada, December 11-14, 2016, Proceedings}, pages
  414--428, 2016.

\bibitem[ABM17]{ABM17}
Georgios Amanatidis, Georgios Birmpas, and Evangelos Markakis.
\newblock On budget-feasible mechanism design for symmetric submodular
  objectives.
\newblock In {\em Web and Internet Economics - 13th International Conference,
  {WINE} 2017, Bangalore, India, December 17-20, 2017, Proceedings}, pages
  1--15, 2017.

\bibitem[AGN14]{AGN14}
Nima Anari, Gagan Goel, and Afshin Nikzad.
\newblock Mechanism design for crowdsourcing: An optimal 1-1/e competitive
  budget-feasible mechanism for large markets.
\newblock In {\em 55th {IEEE} Annual Symposium on Foundations of Computer
  Science, {FOCS} 2014, Philadelphia, PA, USA, October 18-21, 2014}, pages
  266--275, 2014.

\bibitem[AH06]{AggarwalH06}
Gagan Aggarwal and Jason~D. Hartline.
\newblock Knapsack auctions.
\newblock In {\em Proceedings of the Seventeenth Annual {ACM-SIAM} Symposium on
  Discrete Algorithms, {SODA} 2006, Miami, Florida, USA, January 22-26, 2006},
  pages 1083--1092, 2006.

\bibitem[AT02]{ArcherT02}
Aaron Archer and {\'{E}}va Tardos.
\newblock Frugal path mechanisms.
\newblock In {\em Proceedings of the Thirteenth Annual {ACM-SIAM} Symposium on
  Discrete Algorithms, January 6-8, 2002, San Francisco, CA, {USA.}}, pages
  991--999, 2002.

\bibitem[BCGL17]{BCGL17}
Xiaohui Bei, Ning Chen, Nick Gravin, and Pinyan Lu.
\newblock Worst-case mechanism design via bayesian analysis.
\newblock {\em {SIAM} J. Comput.}, 46(4):1428--1448, 2017.

\bibitem[BH16]{BalkanskiH16}
Eric Balkanski and Jason~D. Hartline.
\newblock Bayesian budget feasibility with posted pricing.
\newblock In {\em Proceedings of the 25th International Conference on World
  Wide Web, {WWW} 2016, Montreal, Canada, April 11 - 15, 2016}, pages 189--203,
  2016.

\bibitem[BKS12]{BadanidiyuruKS12}
Ashwinkumar Badanidiyuru, Robert Kleinberg, and Yaron Singer.
\newblock Learning on a budget: posted price mechanisms for online procurement.
\newblock In {\em Proceedings of the 13th {ACM} Conference on Electronic
  Commerce, {EC} 2012, Valencia, Spain, June 4-8, 2012}, pages 128--145, 2012.

\bibitem[BS15]{BalkanskiS15}
Eric Balkanski and Yaron Singer.
\newblock Mechanisms for fair attribution.
\newblock In {\em Proceedings of the Sixteenth {ACM} Conference on Economics
  and Computation, {EC} '15, Portland, OR, USA, June 15-19, 2015}, pages
  529--546, 2015.

\bibitem[CEGP10]{ChenEGP10}
Ning Chen, Edith Elkind, Nick Gravin, and Fedor Petrov.
\newblock Frugal mechanism design via spectral techniques.
\newblock In {\em 51th Annual {IEEE} Symposium on Foundations of Computer
  Science, {FOCS} 2010, October 23-26, 2010, Las Vegas, Nevada, {USA}}, pages
  755--764, 2010.

\bibitem[CGL11]{CGL11}
Ning Chen, Nick Gravin, and Pinyan Lu.
\newblock On the approximability of budget feasible mechanisms.
\newblock In {\em Proceedings of the Twenty-Second Annual {ACM-SIAM} Symposium
  on Discrete Algorithms, {SODA} 2011, San Francisco, California, USA, January
  23-25, 2011}, pages 685--699, 2011.

\bibitem[CK07]{ChenK07}
Ning Chen and Anna~R. Karlin.
\newblock Cheap labor can be expensive.
\newblock In {\em Proceedings of the Eighteenth Annual {ACM-SIAM} Symposium on
  Discrete Algorithms, {SODA} 2007, New Orleans, Louisiana, USA, January 7-9,
  2007}, pages 707--715, 2007.

\bibitem[DPS11]{DPS11}
Shahar Dobzinski, Christos~H. Papadimitriou, and Yaron Singer.
\newblock Mechanisms for complement-free procurement.
\newblock In {\em Proceedings 12th {ACM} Conference on Electronic Commerce
  (EC-2011), San Jose, CA, USA, June 5-9, 2011}, pages 273--282, 2011.

\bibitem[EGG07]{ElkindGG07a}
Edith Elkind, Leslie~Ann Goldberg, and Paul~W. Goldberg.
\newblock Frugality ratios and improved truthful mechanisms for vertex cover.
\newblock In {\em Proceedings 8th {ACM} Conference on Electronic Commerce
  (EC-2007), San Diego, California, USA, June 11-15, 2007}, pages 336--345,
  2007.

\bibitem[ESS04]{ElkindSS04}
Edith Elkind, Amit Sahai, and Kenneth Steiglitz.
\newblock Frugality in path auctions.
\newblock In {\em Proceedings of the Fifteenth Annual {ACM-SIAM} Symposium on
  Discrete Algorithms, {SODA} 2004, New Orleans, Louisiana, USA, January 11-14,
  2004}, pages 701--709, 2004.

\bibitem[GNS14]{GoelNS14}
Gagan Goel, Afshin Nikzad, and Adish Singla.
\newblock Mechanism design for crowdsourcing markets with heterogeneous tasks.
\newblock In {\em Proceedings of the Seconf {AAAI} Conference on Human
  Computation and Crowdsourcing, {HCOMP} 2014, November 2-4, 2014, Pittsburgh,
  Pennsylvania, {USA}}, 2014.

\bibitem[HIM14]{HIM14}
Thibaut Horel, Stratis Ioannidis, and S.~Muthukrishnan.
\newblock Budget feasible mechanisms for experimental design.
\newblock In {\em {LATIN} 2014: Theoretical Informatics - 11th Latin American
  Symposium, Montevideo, Uruguay, March 31 - April 4, 2014. Proceedings}, pages
  719--730, 2014.

\bibitem[HKS18]{HajiaghayiKS18}
Mohammad~Taghi Hajiaghayi, Mohammad~Reza Khani, and Saeed Seddighin.
\newblock Frugal auction design for set systems: Vertex cover and knapsack.
\newblock In {\em Proceedings of the 2018 {ACM} Conference on Economics and
  Computation, Ithaca, NY, USA, June 18-22, 2018}, pages 645--662, 2018.

\bibitem[KKT05]{KarlinKT05}
Anna~R. Karlin, David Kempe, and Tami Tamir.
\newblock Beyond {VCG:} frugality of truthful mechanisms.
\newblock In {\em 46th Annual {IEEE} Symposium on Foundations of Computer
  Science {(FOCS} 2005), 23-25 October 2005, Pittsburgh, PA, USA, Proceedings},
  pages 615--626, 2005.

\bibitem[KSM10]{KempeSM10}
David Kempe, Mahyar Salek, and Cristopher Moore.
\newblock Frugal and truthful auctions for vertex covers, flows and cuts.
\newblock In {\em 51th Annual {IEEE} Symposium on Foundations of Computer
  Science, {FOCS} 2010, October 23-26, 2010, Las Vegas, Nevada, {USA}}, pages
  745--754, 2010.

\bibitem[KT18]{KT18}
Pooya~Jalaly Khalilabadi and {\'{E}}va Tardos.
\newblock Simple and efficient budget feasible mechanisms for monotone
  submodular valuations.
\newblock In {\em Web and Internet Economics - 14th International Conference,
  {WINE} 2018, Oxford, UK, December 15-17, 2018, Proceedings}, pages 246--263,
  2018.

\bibitem[LMSZ17]{LMSZ17}
Stefano Leonardi, Gianpiero Monaco, Piotr Sankowski, and Qiang Zhang.
\newblock Budget feasible mechanisms on matroids.
\newblock In {\em Integer Programming and Combinatorial Optimization - 19th
  International Conference, {IPCO} 2017, Waterloo, ON, Canada, June 26-28,
  2017, Proceedings}, pages 368--379, 2017.

\bibitem[Mye81]{M81}
Roger~B. Myerson.
\newblock Optimal auction design.
\newblock {\em Math. Oper. Res.}, 6(1):58--73, 1981.

\bibitem[Sin10]{S10}
Yaron Singer.
\newblock Budget feasible mechanisms.
\newblock In {\em 51th Annual {IEEE} Symposium on Foundations of Computer
  Science, {FOCS} 2010, October 23-26, 2010, Las Vegas, Nevada, {USA}}, pages
  765--774, 2010.

\bibitem[Sin12]{S12}
Yaron Singer.
\newblock How to win friends and influence people, truthfully: influence
  maximization mechanisms for social networks.
\newblock In {\em Proceedings of the Fifth International Conference on Web
  Search and Web Data Mining, {WSDM} 2012, Seattle, WA, USA, February 8-12,
  2012}, pages 733--742, 2012.

\bibitem[SM13]{SM13}
Yaron Singer and Manas Mittal.
\newblock Pricing mechanisms for crowdsourcing markets.
\newblock In {\em 22nd International World Wide Web Conference, {WWW} '13, Rio
  de Janeiro, Brazil, May 13-17, 2013}, pages 1157--1166, 2013.

\bibitem[Tal03]{Talwar03}
Kunal Talwar.
\newblock The price of truth: Frugality in truthful mechanisms.
\newblock In {\em {STACS} 2003, 20th Annual Symposium on Theoretical Aspects of
  Computer Science, Berlin, Germany, February 27 - March 1, 2003, Proceedings},
  pages 608--619, 2003.

\end{thebibliography}
